\documentclass[american,aps,pra,reprint, superscriptaddress]{revtex4-1}
\usepackage[unicode=true,pdfusetitle, bookmarks=true,bookmarksnumbered=false,bookmarksopen=false, breaklinks=false,pdfborder={0 0 0},backref=false,colorlinks=false] {hyperref}
\hypersetup{ colorlinks,linkcolor=myurlcolor,citecolor=myurlcolor,urlcolor=myurlcolor}
\usepackage{graphics,epstopdf,graphicx,tikz, amsthm, amsmath, amssymb, times, braket, colortbl, color, bm, framed, cleveref, mathrsfs}
\usepackage[up]{subfigure}
\definecolor{myurlcolor}{rgb}{0,0,0.7}

\newcommand{\tinyspace}{\mspace{1mu}}

\newcommand{\proj}[1]{| #1\rangle\!\langle #1 |}

\DeclareMathOperator{\trace}{Tr}
\newcommand{\Ptr}[2]{\trace_{#1}\Pa{#2}}
\newcommand{\Tr}[1]{\Ptr{}{#1}}

\newcommand{\Pa}[1]{\left(#1\right)}

\newcommand{\abs}[1]{\left\lvert\tinyspace #1 \tinyspace\right\rvert}
\newcommand{\norm}[1]{\left\lVert #1 \right\rVert}

\theoremstyle{plain}
\newtheorem{thm}{Theorem}
\newtheorem{lem}[thm]{Lemma}
\newtheorem{prop}[thm]{Proposition}

\newcommand*{\myproofname}{Proof}

\def\ot{\otimes}
\def\complex{\mathbb{C}}
\def\real{\mathbb{R}}

\def\cH{\mathcal{H}}

\makeatother

\begin{document}

 \author{Kaifeng Bu}
 \email{bkf@zju.edn.cn}
 \affiliation{School of Mathematical Sciences, Zhejiang University, Hangzhou 310027, PR~China}

 \author{Heng Fan}
 \email{hfan@iphy.ac.cn}
\affiliation{Beijing National Laboratory for Condensed Matter Physics, Institute of Physics, Chinese Academy of Sciences, Beijing 100190, China}
\affiliation{Collaborative Innovation Center of Quantum Matter, Beijing 100190, China}
\author{Arun Kumar Pati}
 \email{akpati@hri.res.in}
\affiliation{Quantum Information and Computation Group\\
Harish-Chandra Research Institute Chattnag Road, Jhusi,
Allahabad-211019 India}

 \author{Junde Wu}
 \email{wjd@zju.edn.cn}
 \affiliation{School of Mathematical Sciences, Zhejiang University, Hangzhou 310027, PR~China}
%

%\title{Quantifying PT-Asymmetry }
%\title{Resource Theory of ${\cal PT}$-Asymmetry }

\title{Resource Theory of Special Antiunitary Asymmetry }

\begin{abstract}
We propose the resource theory of a special antiunitary asymmetry in quantum theory. The notion of antiunitary asymmetry, in particular, ${\cal PT}$-asymmetry
is  different from the
usual resource theory for asymmetry about unitary representation of
a symmetry group, as the ${\cal PT}$ operator is an antiunitary operator with ${\cal P}$ being any self-inverse unitary and ${\cal T}$ being the time-reversal operations.
Here, we introduce the
${\cal PT}$-symmetric states, ${\cal PT}$-covariant operations and ${\cal PT}$-asymmetry measures. For single qubit system, we find duality relations
between the ${\cal PT}$-asymmetry measures and the coherence. Moreover, for two-qubit states
we prove the duality relations between the ${\cal PT}$-asymmetry measures and entanglement measure such as the concurrence. This gives a
resource theoretic interpretation to the concurrence which is lacking till today. Thus, the ${\cal PT}$-asymmetry measure and entanglement
can be viewed as two sides of an underlying resource.
Finally, the ${\cal PT}$-symmetric dynamics is discussed and some open questions are addressed.
\end{abstract}

\maketitle

% %================================%
%\section{Introduction}
% %================================%

{\it Introduction.--} Quantum resource theories \cite{Oppenheim13, FBrandao15} have played a pivotal role in the development and quantitative understanding of various physical phenomena
in quantum physics and quantum information theory.  A resource theory consists of two basic elements: free operations and free states. Any operation (or state) is  dubbed as
a resource if it falls out of the set of free operations (or the set of free states).  The most significant resource theory is entanglement \cite{HorodeckiRMP09}, which is
a basic resource for various quantum information processing protocols, such as
the superdense coding \cite{Bennett1992}, teleportation \cite{Bennett1993} and remote state preparation \cite{pati2000}.
The other notable examples include the resource theories of thermodynamics \cite{Fernando2013},
asymmetry \cite{Bartlett2007,Gour2008,Gour2009,Marvian2012,Marvian2013,Marvian14,Marvian2014}, coherence
\cite{Baumgratz2014,Girolami2014,Streltsov2015,Winter2016,Killoran2016,Napoli2016,Chitambar2016,Chitambar2016a} and steering \cite{Rodrigo2015}. The main advantages of
having a resource theory for some physical quantity are the succinct understanding of various physical processes and operational quantification of the relevant resources at ones
disposal.

The Hermiticity is one fundamental requirement of quantum mechanics for the Hamiltonian
of a quantum system, which guarantees that the energies are real and the total probability of the quantum state
is conserved during the evolution of the system. However, it has been proved that a broad class of non-Hermitian Hamiltonian
 with ${\cal PT}$-symmetry can also have real spectra and probability conservation by redefined
 inner product \cite{Bender1998,Bender1999,Bender2002,Bender2004,Bender2006,Bender2007}, where $\mathcal{P}$ denotes the parity operator and $\mathcal{T}$ denotes the time reversal
 operator. This implies that ${\cal PT}$-symmetric theory constitutes a complex generalization
 of conventional quantum mechanics \cite{Bender2002}.
Moreover, in the system with ${\cal PT}$-symmetric
non-Hermitian Hamiltonian, a number of interesting phenomena and applications
appear in both  classical and quantum regimes, such as undirectional invisibility \cite{Lin2011,Feng2013,Yin2013},
non-Hermitian
Bloch oscillation \cite{Longhi2009,Martin2015},
perfect laser absorbers \cite{Chong2010,Longhi2010,Longhi2014},
ultrafast quantum state transformation \cite{Bender2007PRL}, quantum state discrimination with single-shot measurement \cite{Bender2012} and the potential violation
of the no-signalling principle \cite{Lee2014,Li2016}.
However, most research  focus on  the  ${\cal PT}$-symmetric Hamiltonian,  never
consider the quantum state with ${\cal PT}$-symmetry. Thus,
 the following questions arise: how to define
a  ${\cal PT}$-symmetric quantum state, what is the physical meaning of ${\cal PT}$-symmetric
 states and how to define measures of ${\cal PT}$-asymmetry.

Recently, the quantification of time reversal asymmetry  \cite{Gour2009JMP} and  $\mathcal{CPT}$ asymmetry \cite{Skotiniotis2013,Skotiniotis2014} have been considered in
antiunitary  and unitary representations, respectively. However, there still remains a question
in which representation to choose those relevant operations \cite{Piotr2015}.
%The difference between these two representation
%is similar to  the difference between rotation about $180^\circ$ and reflection.
In this work, we use the
framework of quantum resource theory to quantify ${\cal PT}$-asymmetry and investigate the relationship
between ${\cal PT}$-asymmetry measures, quantum coherence and entanglement. Note that here ${\cal P}$ is a self-inverse unitary operator
(need not be parity operator) and ${\cal T}$ is time-reversal operator.
 ${\cal PT}$ operator can be realized as  a special kind of  antiunitary operator \cite{Porter1965,Bender2002JPA}, which
 is in contrast with
the resource theory of asymmetry  on
the unitary representation of a symmetric group \cite{Bartlett2007,Gour2008,Gour2009,Marvian2012,Marvian2013,Marvian14,Marvian2014}.
Thus, the resource of ${\cal PT}$-asymmetry will be a special
kind of resource theory  of antiunitary asymmetry.
Though we cannot tensor the antilinear operator with the identity operator consistently, because
antilinear operators are nonlocal, nevertheless they have been used to
measure entanglement of a given bipartite state \cite{Hill1997,Wootters1998,Uhlmann2000,Coffman2000,Osterloh2005,Osborne2006}. And there is
a famous entanglement measure--the concurrence \cite{Wootters1998}, which is indeed
constructed from antilinear operators. It is quite satisfying that the resource theory of antiunitary asymmetry provides a
unified view of two fundamental
resources of quantum world such as the coherence and entanglement.
For single qubit, we reveal a duality relation between the ${\cal PT}$-asymmetry measure and the coherence.
For two-qubit pure states, we prove duality relations between the ${\cal PT}$-asymmetry measures and entanglement measure such as the
concurrence.  Amazingly, we find that the pure
bipartite state is maximally entangled if and only if it is
${\cal PT}$-symmetric. Therefore, entanglement is a special ${\cal PT}$-symmetry in some sense.
Furthermore, as $\mathcal{K}=*$ is an unphysical operator, that is $\mathcal{K}$ cannot be realized
in a physical system, then it is hard to calculate the ${\cal PT}$-asymmetry measure. However,we show that
via the embedding quantum simulator \cite{Casanova2011,Candia2013,Georgescu2014,Zhang2015a,Loredo2016,Chen2016}, the ${\cal PT}$ asymmetry measure can be calculated
efficiently.

%============================================================%
%\section{The measure of PT symmetry for a state}\label{sec:PT_asm_pre}
%============================================================%

\noindent
\emph{~~ ${\cal PT}$-symmetric state. --}
 Consider the self-inverse unitary operator $\mathcal{P}$ and  time reversal operator $\mathcal{T}$, where
$\mathcal{P}$ and $\mathcal{T }$ satisfy the following condition:(1) $\mathcal{P}=\mathcal{P}^\dag$, $\mathcal{P}^2=I$,
(2) $\mathcal{T}=U\mathcal{K}$, where $U$ is a unitary operator with $U=U^t$ and $\mathcal{K}=*$ is the complex conjugation and
 (3) $[\mathcal{P},\mathcal{T}]=0$.
Note that any antiunitary operator $\Theta$ with $\Theta= \Theta^\dag=\Theta^{-1}$ can be written in the form $V\mathcal{K}$,
where $V$ is a unitary operator with $V=V^t$ and $\mathcal{K}$ is the complex conjugation with respect
to a given basis. Such antiunitary operator is called conjugation and plays an important role in quantum information
 theory \cite{Uhlmann2000,Uhlmann2016}. It is easy to
see that such conjugation is equivalent to the ${\cal PT}$ operator defined above. Thus, the resource  theory of
${\cal PT}$-asymmetry considered in this work is  a special kind of anitunitary asymmetry resource theory and
may indicate the way towards formulating the general resource theory of anitunitary asymmetry.

Throughout this paper, we assume that self-inverse unitary operator and time reversal operator
always satisfy these conditions.
Given a quantum state $\rho$, once we apply the operations $\mathcal{P}$ and $\mathcal{T}$, the final state will
be $\mathcal{PT}\rho \mathcal{PT}$ (Since $\mathcal{PT}\rho \mathcal{PT}=\mathcal{P}U\mathcal{K}\rho \mathcal{K}U^* \mathcal{P}=\mathcal{P}U\rho^* U^\dag \mathcal{P}$ is a quantum state).
If the initial state is equal to the final state, that is $[\rho, {\cal PT}]=0$, then we call the state $\rho$ is ${\cal PT}$-symmetric
state. If the initial state is not equal to the final state, that is $[\rho, {\cal PT}]\neq0$, then we call the state $\rho$
is ${\cal PT}$-asymmetric. Moreover, we denote the set of all ${\cal PT}$-symmetric states by $Sym(\mathcal{P},\mathcal{T})$.

\noindent
\emph{~~ ${\cal PT}$-covariant operation. --}
To characterize the quantum operation which transform the ${\cal PT}$-symmetric states to the
 ${\cal PT}$-symmetric states, we distinguish quantum operations with and without subselection.
Any quantum operation  $\Phi$ can be described using a set of Kraus operators $\set{K_\mu}$ with
$\Phi(\cdot)=\sum_\mu K_\mu(\cdot) K^\dag_\mu$.
The operation $\Phi: \mathcal{D}(\cH)\to \mathcal{D}(\cH)$ is called ${\cal PT}$-covariant if $\Phi(\mathcal{PT}(\cdot)\mathcal{PT})=\mathcal{PT}\Phi(\cdot)\mathcal{PT}$,
that is $[\Phi, \mathcal{PT}]=0$. Such operations are denoted by  $\Phi_{\mathcal{PT}_{CO}}$. (Of course, we can also consider the ${\cal PT}$-covariant operation with
different $\mathcal{PT}$, that is $\Phi(\mathcal{P}_1\mathcal{T}_1(\cdot)\mathcal{P}_1\mathcal{T}_1)=\mathcal{P}_2\mathcal{T}_2\Phi(\cdot)\mathcal{P}_2\mathcal{T}_2$.)
Besides this, we also need to consider the quantum operations with subselection. Thus, a quantum operation
$\Phi$ is called selective $\mathcal{PT}$-covariant if the Kraus operators $\set{K_\mu}$ of $\Phi$ satisfy
$K_\mu(\mathcal{PT}(\cdot)\mathcal{PT})K^\dag_\mu=\mathcal{PT}K_\mu(\cdot)K^\dag_\mu \mathcal{PT}$ for any $\mu$.
%Generally, we consider  the case: $P_1=P_2$, $T_1=T_2$. And in this case, we just call the above two
%PT-covariant operation as PT-covariant and selective PT-covariant for simplicity.

\noindent
\emph{~~ The measure of ${\cal PT}$-asymmetry for a state. --}
When the state is ${\cal PT}$-asymmetric, that is it breaks the ${\cal PT}$-symmetry, we want to
quantify how much the ${\cal PT}$-symmetry is broken by the given state. Thus, we need to introduce the
${\cal PT}$ asymmetry measure, like the entanglement measure \cite{Vedral1998,Plenio2007}, asymmetry measure \cite{Marvian2014,Marvian2012} and
coherence measure \cite{Baumgratz2014,Girolami2014}. Now, we list the conditions that any function $\Gamma$ from  a state to a real
number needs to satisfy in order to be a proper ${\cal PT}$-asymmetry measure.

For any proper ${\cal PT}$-asymmetry measure $\Gamma$, it needs to satisfy the following conditions:

(C1) $\Gamma(\rho, \mathcal{ PT})=0$ iff $[\rho, \mathcal{PT}]=0$.

(C2) Monotone under ${\cal PT}$-covariant operations  $\Phi_{\mathcal{PT}_{CO}}$, that is
$\Gamma(\Phi_{\mathcal{PT}_{CO}}(\rho),\mathcal{PT})\leq \Gamma(\rho, \mathcal{PT})$.

(C2') Monotone under selective ${\cal PT}$-covariant operations:
$\sum_\mu p_\mu \Gamma(\rho_\mu, \mathcal{PT})\leq \Gamma(\rho,\mathcal{PT})$,
where $K_\mu(\mathcal{PT}(\cdot)\mathcal{PT})K^\dag_\mu=\mathcal{PT}K_\mu(\cdot)K^\dag_\mu \mathcal{PT}$ and
$\rho_\mu=K_\mu\rho K^\dag_\mu/p_\mu$ with $p_\mu=\mathrm{Tr}(K_\mu \rho K^\dag_\mu)$.

(C3) Convexity: $\Gamma(\sum_n p_n \rho_n, \mathcal{PT})\leq \sum_n p_n\Gamma(\rho_n, \mathcal{PT})$, where
$\set{\rho_n}$ is a set of states and $p_n\geq 0$ with $\sum_np_n=1$.

The condition (C1) means the ${\cal PT}$-asymmetry measure vanishes if and only if this state is
${\cal PT}$-symmetric. We can weaken this condition as (C1'): $\Gamma(\rho, \mathcal{PT})=0$ if $[\rho, \mathcal{PT}]=0$.
Naturally,  ${\cal PT}$ -asymmetry measure cannot increase  under the ${\cal PT}$-covariant operations, thus the conditions (C2) is necessary.
Furthermore, the condition $(C2')$ implies that the average ${\cal PT}$-asymmetry after the ${\cal PT}$-covariant operations
with subselection cannot be greater than the ${\cal PT}$-asymmetry of the initial state. This condition
may be important in real experiment so that we list this condition here.
The condition (C3) is the convexity of the ${\cal PT}$-asymmetry measure, which is  the requirement of any proper
asymmetry monotone \cite{Marvian2012}.

We now give several ${\cal PT}$-asymmetry measures via the relative entropy, the skew information and the fidelity measures.

\noindent
\emph{~~ Relative entropy of ${\cal PT}$-asymmetry. --}
The quantum relative entropy for states $\rho$ and $\sigma$ is defined as $S(\rho||\sigma):= \Tr{\rho\log\rho}-\Tr{\rho\log \sigma}$.
The relative entropy of
$\mathcal{PT}$-asymmetry measure $\Gamma_r$ is defined as
\begin{eqnarray}
\Gamma_r(\rho, \mathcal{PT})=\min_{\sigma\in Sym(\mathcal{P}, \mathcal{T})}S(\rho||\sigma).
\end{eqnarray}
First, we get a closed form expression of $\Gamma_r$ to avoid the
minimization and it is given by
\begin{eqnarray}
\Gamma_r(\rho, \mathcal{PT})=S(\rho||\rho^{\mathcal{PT}})=S(\rho^{\mathcal{PT}})-S(\rho),
\end{eqnarray}
where $\rho^{\mathcal{PT}}=\frac{1}{2}(\rho+\mathcal{PT}\rho \mathcal{PT})$ is ${\cal PT}$-symmetric and
$\Gamma_r$ fulfills the conditions (C1), (C2), (C2') and (C3) as
a proper ${\cal PT}$-asymmetry measure (See Appendix \ref{ap:Re_pt}). Then, for any state $\rho$,
$\Gamma_r(\rho, \mathcal{PT})=S(\rho^{\mathcal{PT}})-S(\rho)\leq 1$, as
$S(\sum_i p_i\rho_i)\leq \sum_i p_iS(\rho_i)+H(\set{p_i})$ \cite{Nielsen10} and $S(\mathcal{PT}\rho \mathcal{PT})=S(\rho)$.
Since $S(\sum_i p_i\rho_i)= \sum_i p_iS(\rho_i)+H(\set{p_i})$ is equivalent to that  $\rho_i$ have orthogonal
supports \cite{Nielsen10}, then $\Gamma_r(\rho, \mathcal{PT})=1$ iff $\rho \bot \mathcal{PT}\rho \mathcal{PT}$.

\noindent
\emph{~~ Skew information of ${\cal PT}$-asymmetry. --}
Let us define the skew information of ${\cal PT}$-asymmetry  $\Gamma_s$ as
\begin{eqnarray}
\Gamma_s (\rho, \mathcal{PT}):&=&-\frac{1}{2}\Tr{[\rho^{1/2},\mathcal{PT}]^2}\nonumber\\
&=&1-\Tr{\rho^{1/2}\mathcal{PT}\rho^{1/2}\mathcal{PT}},
\end{eqnarray}
where $[\cdot,\cdot]$ denote the
commutator and $[\rho^{1/2}, \mathcal{PT}]^2=\rho+\mathcal{PT}\rho \mathcal{PT}-2\rho^{1/2}\mathcal{PT}\rho^{1/2}\mathcal{PT}$.
Note that, in the definition of Wigner-Yanase-Dyson skew information $I(\rho, O)=-\frac{1}{2}\Tr{[\rho^{1/2},O]}$,
the operator $O$ is required to be an observable \cite{Wigner1963}, that is $O$ must be a Hermitian, however
${\cal PT}$ is not a linear operator, thus ${\cal PT}$ is not an observable. Therefore, we cannot use the
properties of skew information to state that $\Gamma_s$ satisfy the conditions
(C1), (C2), (C2') and (C3). However, $\Gamma_s$ still fulfills   these conditions (See Appendix \ref{ap:Sk_pt}).
Obviously, for any state $\rho$, $\Gamma_s(\rho, PT)\leq 1$ and the equality holds iff $\rho \bot \mathcal{PT}\rho \mathcal{PT}$.

The above two quantities $\Gamma_r$ and $\Gamma_s$ are proper
${\cal PT}$-asymmetry measures ( An example is presented in Fig.\ref{fig1}). Of course, there may be other
possible ${\cal PT}$-asymmetry measure, like the
${\cal PT}$-asymmetry measure induced by the trace norm, the Hilbert Schmidt norm and so on.
Here, we introduce another interesting ${\cal PT}$-asymmetry measure defined by
the fidelity.

\noindent
\emph{~~ Fidelity measure of ${\cal PT}$-asymmetry. --}
Let us consider the ${\cal PT}$-asymmetry measure defined as
\begin{eqnarray}
\Gamma_F(\rho, \mathcal{PT})
&=&1-F(\rho, \mathcal{PT}\rho \mathcal{PT})\nonumber\\
&=&1-\Tr{\sqrt{\sqrt{\rho}\mathcal{PT}\rho \mathcal{PT}\sqrt{\rho}}},
\end{eqnarray}
which fulfils the conditions
(C1), (C2), (C2') and (C3) (See Appendix \ref{ap:F_pt}).

Based on the proof in Proposition \ref{prop:skewm} (See Appendix \ref{ap:Sk_pt}), we have $\Gamma_F\leq \Gamma_s$. Following
Ref.\cite{Uhlmann2000},
$\Gamma_F(\rho, \mathcal{PT})$ can be written as
\begin{eqnarray}\label{eq:deco_F}
\Gamma_F(\rho, \mathcal{PT})=\min \sum_k p_k\Gamma_F(\psi_k, \mathcal{PT}),
\end{eqnarray}
where the minimum is taken over all the decomposition of $\rho=\sum_k p_k\proj{\psi_k}$.
Furthermore, the optimal decomposition can be found in Ref.\cite{Uhlmann2000}.

\begin{figure}
\label{fig1}
\includegraphics[width=80mm]{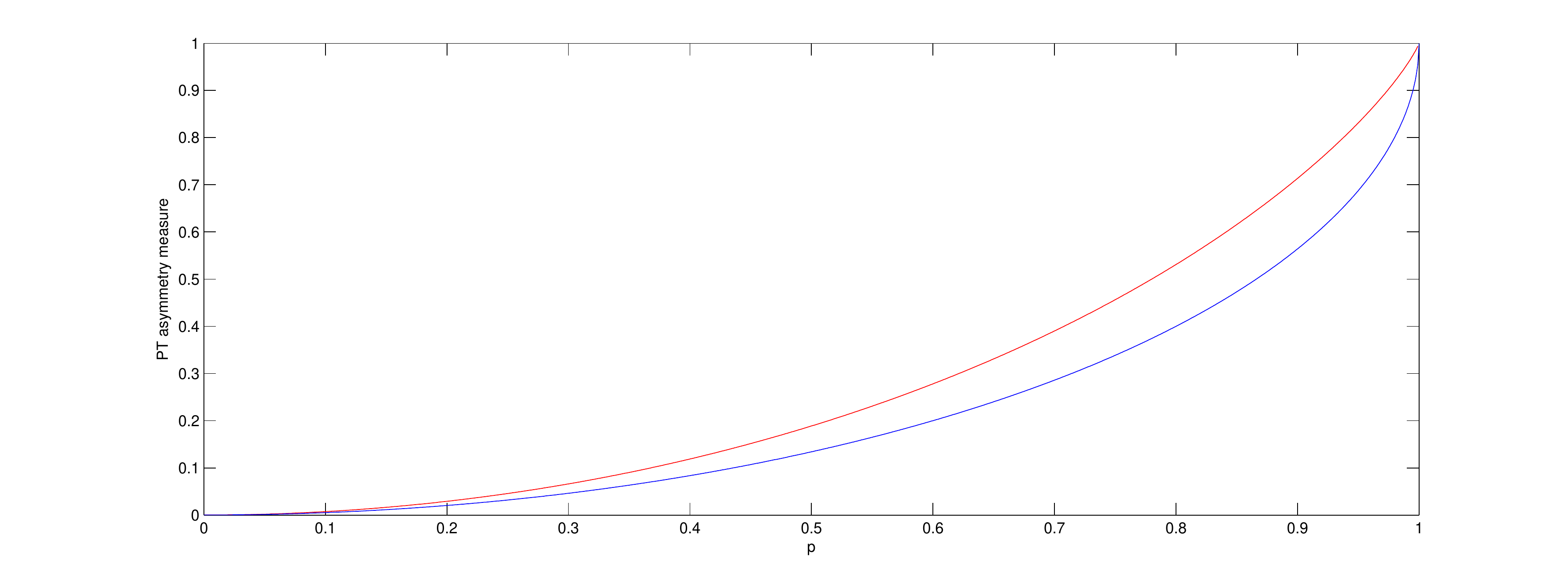}
\caption{The plot shows the ${\cal PT}$-asymmetry measure $\Gamma_r$ (red line) and $\Gamma_s$ (blue line) for qubit states
$\rho=(1-p)\mathbb{I}/2+p\proj{\psi}$ with $\ket{\psi}=\frac{1}{\sqrt{2}}(\ket{0}+\ket{1})$, $p\in[0,1]$ under the unitary operator $\mathcal{P}=\sigma_z=\left(
\begin{array}{cc}
1 & 0\\
0 & -1
\end{array}
\right)$ and time reversal operator $\mathcal{T}=*$  }
\label{fig1}
\end{figure}

%============================================================%
%\section{Complementarity of PT-Asymmetry and Entanglement}\label{sec:PTvsEnt}
%============================================================%

{\it Duality of ${\cal PT}$-Asymmetry, Coherence and Entanglement.--}\label{sec:PTvsEnt}
Given a self-inverse unitary operator $\mathcal{P}$ and a time reversal operator $\mathcal{T}$, for any pure state
$\rho=\proj{\psi}$, we have
\begin{eqnarray}\label{eq:psk}
\Gamma_s (\psi, \mathcal{PT})&=&1-|\bra{\psi}\mathcal{PT}\ket{\psi}|^2,\\
and~~~\Gamma_F (\psi, \mathcal{PT})&=&1-|\bra{\psi}\mathcal{PT}\ket{\psi}|.
\end{eqnarray}
 This can be interpreted  as follows.
Imagine that we have two copies of a pure state $\ket{\psi}$, and one is rotated in space by a unitary
operator $\mathcal{P}$ and the other is transformed under time reversal operator $\mathcal{T}$. The final states
will be $\mathcal{P}\ket{\psi}$ and $\mathcal{T}\ket{\psi}$, and we want to know whether these final
states coincide or not. If they coincide, then this means that the effect of the parity operator
and the time reversal operator leaves the state $\ket{\psi}$ invariant, and
we say  $\ket{\psi}$ has ${\cal PT}$-symmetry. Otherwise, the state $\ket{\psi}$
breaks the ${\cal PT}$-symmetry.

Moreover, the spectrum of $\rho^{\mathcal{PT}}$ with $\rho=\proj{\psi}$ is $\set{\frac{1}{2}-\frac{1}{2}|\bra{\psi}\mathcal{PT}\ket{\psi}|,
\frac{1}{2}+\frac{1}{2}|\bra{\psi}\mathcal{PT}\ket{\psi}|}$. Thus,
\begin{eqnarray}\label{eq:pre}
\Gamma_r (\psi, \mathcal{PT})=H\Pa{\frac{1}{2}-\frac{1}{2}|\bra{\psi}\mathcal{PT}\ket{\psi}|}
\end{eqnarray}
where $H(p)=-p\log(p)-(1-p)\log(1-p)$ is the Shannon entropy for the probability
distribution $\set{p,1-p}$.

Let us consider the simplest case: a single qubit system.
We take  $\mathcal{P}=\sigma_x=\left( \begin{array}{ccc}
0& 1\\
1  & 0
\end{array} \right)$ and $\mathcal{T}=*$.
Then for pure qubit state $\ket{\psi}=(\psi_1,\psi_2)^t$, where $t$ denotes the transpose,
\begin{eqnarray*}
\Gamma_s(\psi, \mathcal{PT})&=&1-|\bra{\psi}\mathcal{PT}\ket{\psi}|^2\\
&=&1-|\psi_1\psi_2+\psi_2\psi_1|^2\\
&=&1-4|\psi_1|^2|\psi_2|^2,
\end{eqnarray*}
and
\begin{eqnarray*}
\Gamma_r(\psi, \mathcal{PT})=H(\frac{1}{2}-2|\psi_1|^2|\psi_2|^2).
\end{eqnarray*}

Thus, a pure state $\psi$  is ${\cal PT}$-symmetric iff
$|\psi_1|=|\psi_2|=1/\sqrt{2}$. This suggests that there should be a connection between
the quantum coherence \cite{Baumgratz2014} and the ${\cal PT}$-asymmetry measure. In fact, for a single qubit system we have  duality
relations between the $l_1$-norm of coherence and the ${\cal PT}$-asymmetry measures as  given by
\begin{eqnarray}
\Gamma_s(\psi, \mathcal{PT}) + C_{l_1}(\psi)^2 = 1,\nonumber\\
\Gamma_F(\psi, \mathcal{PT}) + C_{l_1}(\psi) = 1,
\end{eqnarray}
where $C_{l_1}(\psi) = \sum_{i\not= j} |\rho_{ij}| =  2 |\psi_1||\psi_2|$ is the $l_1$-norm of
coherence for a single qubit. Therefore, a maximally pure coherent state is actually a ${\cal PT}$-symmetric state.

However, in two-qubit system, we have two different ways to consider the ${\cal PT}$-asymmetry.
On the one hand, we can  construct the $\mathcal{P}$, $\mathcal{T}$ operators on 2-qubit system using $\mathcal{P}$, $\mathcal{T}$ operators on single qubit systems
like $\mathcal{P}_1\mathcal{T}_1\ot \mathcal{P}_2\mathcal{T}_2$. On the other  hand, we can construct $\mathcal{P}$, $\mathcal{T}$ operators on 2-qubit system which cannot
be constructed from single qubit systems, and this may be connected with entanglement closely.

For two-qubit  pure state $\ket{\Psi}$ the famous entanglement monotone-- the concurrence \cite{Wootters1998} is defined
as
\begin{eqnarray}\label{eq:defc}
C(\Psi)=|\bra{\Psi}\sigma_y \ot \sigma_y \mathcal{K}\ket{\Psi}|.
\end{eqnarray}
Now, we prove duality relations between the ${\cal PT}$-asymmetry measures and the concurrence.
Using the definitions of $\Gamma_s$, $\Gamma_r$ and $C(\psi)$ for any pure two qubit state, we have the following theorem.
\begin{thm}
Given a two-qubit system with self-inverse unitary operator $\mathcal{P}=\sigma_y\ot \sigma_y$
and time reversal operator $\mathcal{T}=*$,  for pure bipartite states $\ket{\Psi}$ we have
\begin{eqnarray}\label{eq:skvse}
\Gamma_s(\Psi, \mathcal{PT})+C(\Psi)^2=1,\\
\label{eq:fvse}
\Gamma_F(\Psi, \mathcal{PT})+C(\Psi)=1,
\end{eqnarray}
and
\begin{eqnarray}\label{eq:rvse}
\Gamma_r (\Psi,\mathcal{PT})=H\Pa{\frac{1}{2}-\frac{1}{2}C(\Psi)},
\end{eqnarray}
where $H(p)=-p\log(p)-(1-p)\log(1-p)$ is the Shannon entropy for the probability
distribution $\set{p,1-p}$ and $C(\Psi)$ is the concurrence for pure state $\ket{\Psi}$.

For any two-qubit mixed states $\rho$, the equalities  may not hold. However, we still have the
following inequality:
\begin{eqnarray}
\label{ineq:skvse}\Gamma_s(\rho,\mathcal{PT})+C(\rho)^2\leq1,\\
\label{ineq:fvse}\Gamma_F(\rho,\mathcal{PT})+C(\rho)\leq 1,\\
\label{ineq:rvse} \Gamma_r(\rho,\mathcal{PT})\leq H(\frac{1}{2}-\frac{1}{2}C(\rho)),
\end{eqnarray}
where $C(\rho)=\min \sum_k p_k C(\Psi_k)$ and
the minimum is taken over all the pure states decomposition of $\rho=\sum_k p_k\proj{\Psi_k}$ \cite{Wootters1998,Uhlmann2000}.

In fact, we can prove the following
\begin{eqnarray}\label{eq:coa}
\Gamma_F(\rho,\mathcal{PT})+CoA(\rho)=1,
\end{eqnarray}
where the concurrence of assistance $CoA(\rho)=\max \sum_k p_k C(\Psi_k)$ and the maximum is taken over all the
 pure states decomposition of $\rho=\sum_k p_k\proj{\Psi_k}$ \cite{Laustsen2003,Gour2005}.
 \end{thm}
The proof of this theorem is presented in the Appendix \ref{ap:PTvsEnt}.

Since the concurrence quantifies the entanglement of a pure bipartite state, thus
the above proposition shows that  a pure state has more $\mathcal{PT}$-symmetry
with $\mathcal{P}=\sigma_y\ot \sigma_y$ and $\mathcal{T}=*$ if and only if this state is more entangled, i.e.,
the pure state $\Psi$ is a $\mathcal{PT}$-symmetric state  iff  $\Psi$ is maximally entangled.
Therefore, entanglement is  a special kind of $\mathcal{P}\mathcal{T}$-symmetry in some sense.
Our formalism, the resource theory of antiunitary asymmetry, in fact, provides a unified view of two fundamental resources such as the
quantum coherence and the entanglement.

To generalize these notions, we consider the relationship between $\mathcal{P}\mathcal{T}$-asymmetry and
entanglement in multi-qubit  system.
In N-qubit system with $N\geq 2$, there exists a systematic procedure
to define entanglement monotone for pure states via three operational
building blocks \cite{Osterloh2005,Candia2013}: $\mathcal{K}=*$, $\sigma_y$ and $g_{ij}\sigma_i\sigma_j$,  where
$g_{ij}=diag\set{-1,1,0,1}$, $\sigma_0=I$, $\sigma_1=\sigma_x$,
$\sigma_2=\sigma_y$ ad $\sigma_3=\sigma_z$.
If N is even, then the simplest entanglement monotone
is $|\bra{\Psi}\sigma^{\ot N}_y\mathcal{K}\ket{\Psi}|$
and if N is odd, $|\sum_{ij}g_{ij}\bra{\Psi}\sigma_i\ot\sigma^{\ot N-1}_y\mathcal{K}\ket{\Psi}\bra{\Psi}\sigma_j\ot\sigma^{\ot N-1}_y\mathcal{K}\ket{\Psi}|$ is
the simplest entanglement monotone \cite{Osterloh2005,Candia2013}. With this construction, for N=2, we have the entanglement
monotone- the concurrence: $C(\Psi)=|\bra{\Psi}\sigma_y\ot \sigma_y \mathcal{K}\ket{\Psi}|$.
For N=3, we have the 3-tangle \cite{Coffman2000} defined as
$\tau_3(\Psi):=|\sum_{ij}g_{ij}\bra{\Psi}\sigma_i\ot\sigma^{\ot 2}_y\mathcal{K}\ket{\Psi}\bra{\Psi}\sigma_j\ot\sigma^{\ot 2}_y\mathcal{K}\ket{\Psi}|$.

Thus, we have defined entanglement monotone $\tau_N$ for any $N$-qubit system and it is easy to see that for the even and odd
cases we have the following relations:

(i) if $N=2k$, then
\begin{eqnarray}
\Gamma_s(\Psi, \mathcal{P}_2\mathcal{T})&+&\tau_N(\Psi)^2=1,\\
\Gamma_r (\Psi,\mathcal{P}_2\mathcal{T})&=& H\Pa{\frac{1}{2}-\frac{1}{2}\tau_N(\Psi)},
\end{eqnarray}
where  $\mathcal{P}_2=\sigma^{\ot N}_y$, $\mathcal{T}=*$ and $H(p)=-p\log(p)-(1-p)\log(1-p)$ is the Shannon entropy for the probability
distribution $\set{p,1-p}$.

(ii) if $N=2k+1$, then
\begin{eqnarray}
\tau_N(\Psi)&=&|-\Gamma_s(\Psi,\mathcal{P}_0\mathcal{T})+\Gamma_s(\Psi, \mathcal{P}_2\mathcal{T})\nonumber\\
&&+\Gamma_s(\Psi,\mathcal{P}_3\mathcal{T})-1|,
\end{eqnarray}
where $\mathcal{P}_i=\sigma_i\ot \sigma^{\ot N-1}_y$ and $\mathcal{T}=*$.

Since $\mathcal{K}=*$ is an unphysical operator, one may think that
we need to perform full tomography to calculate $\bra{\Psi}\mathcal{PT}\ket{\Psi}=\bra{\Psi}\mathcal{P}U\mathcal{K}\ket{\Psi}$ \cite{Casanova2011,Candia2013}.
However, based on the embedding quantum simulator (EQS), such quantity
can be calculated efficiently \cite{Casanova2011,Candia2013}.
This technique of embedding quantum simulator \cite{Casanova2011,Candia2013,Georgescu2014,Zhang2015a,Loredo2016,Chen2016}  is described as follows:
define a mapping $\mathcal{M}:\complex^{d}\rightarrow \real^{2d}$
as
\begin{equation}
\ket{\Psi}=\left(\begin{array}{cccccccc}
\psi^1_{re}+i\psi^1_{im}\\
\psi^2_{re}+i\psi^2_{im}\\
\psi^3_{re}+i\psi^3_{im}\\
\vdots
\end{array}
\right)
\longrightarrow
 \ket{\widetilde{\Psi}}=\left(
 \begin{array}{ccccccccccc}
 \psi^1_{re}\\ \psi^2_{re}\\
 \psi^3_{re}\\ \vdots \\
 \psi^1_{im}\\ \psi^2_{im}\\
 \psi^3_{im}\\ \vdots
 \end{array}
 \right)
\end{equation}
The reverse mapping is given by $\ket{\Psi}=M\ket{\widetilde{\Psi}}$,
with $M=(1,i)\ot I_{d}$ and $\mathcal{K}\ket{\Psi}=M (\sigma_z\ot I_{d})\ket{\widetilde{\Psi}}$.
Thus, one has
\begin{eqnarray}
\bra{\Psi}\mathcal{P}U\mathcal{K}\ket{\Psi}
=\bra{\widetilde{\Psi}}M^\dag  \mathcal{P}UM(\sigma_z\ot I_{d})\ket{\widetilde{\Psi}}.
\end{eqnarray}
By the embedding quantum simulator, the quantity $|\bra{\Psi}\mathcal{PT}\ket{\Psi}|$ can be
calculated efficiently, which means the $\mathcal{PT}$-asymmetry measures such as $\Gamma_r, \Gamma_s, \Gamma_F$
can be calculated efficiently (see Fig. \ref{fig2} ).

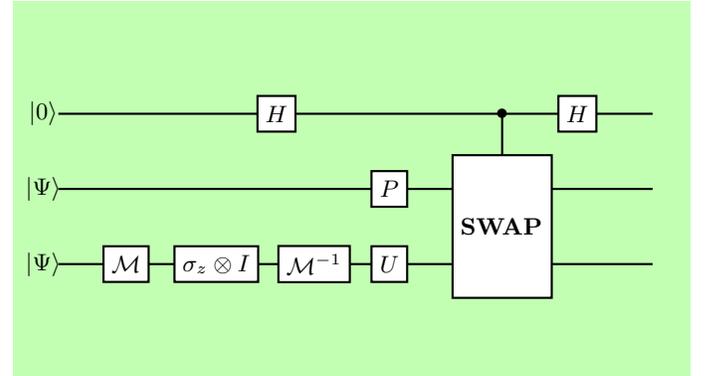
\begin{figure}[ht]
\centering
 \begin{tikzpicture}[thick]
 \fill[green!80!yellow!30](-0.5,3) rectangle (8.5,-2);
    \tikzstyle{operator} = [draw,fill=white,minimum size=1.5em]
  %  \tikzstyle{phase} = [fill,shape=circle,minimum size=5pt,inner sep=0pt]
    \tikzstyle{surround} = [fill=blue!10,thick,draw=black,rounded corners=2mm]
    \tikzstyle{block} = [rectangle, draw, fill=white,
    text width=3.5em, text centered, , minimum height=6em]
    %
    % Qubits
    \node at (-0.1,1.5) (q0){$\ket{0}$};
    \node at (-0.1,0.5) (q0){$\ket{\Psi}$};
    \node at (-0.1,-0.5) (q1){$\ket{\Psi}$};
    \node at (6,1.5) (q5) {$\bullet$};
    \draw[-] (0.1,1.5)  -- (8,1.5);
   \draw[-] (0.1,0.5)  -- (8,0.5);
   \draw[-] (0.1,-0.5) -- (8,-0.5);
   \draw[-] (6,1.5) -- (6,-0.5);
    %\draw[dashed,blue] (0.5, 0.2) rectangle (4.1,-1.2);
   %\node at (2.3,-1.3) (q3) {EQS};
    \node[operator] (op1) at (4.5,0.5) {$P$};
    \node[operator] (op2) at (4.5,-0.5) {$U$};
    \node[operator] (op3) at (1,-0.5) {$\mathcal{M}$};
    \node[operator] (op4) at (2.2,-0.5) {$\sigma_z\ot I$};
    \node[operator] (op5) at (3.5,-0.5) {$\mathcal{M}^{-1}$};
    \node[block] (b1) at (6,0) {$\mathrm{\mathbf{SWAP}}~$};
    \node[operator] (op6) at (3,1.5) {$H$};
    \node[operator] (op7) at (7,1.5) {$H$};

    \end{tikzpicture}
\caption{A quantum network for estimation  of $\mathcal{PT}$-asymmetry. The probability
of finding the control qubit (the top line) in state $\ket{0}$: $p_0$ depends on
the $\mathcal{PT}$ asymmetry of $\ket{\Psi}$ \cite{Ekert2002,Erik2000}, that is $p_0=(1+|\bra{\Psi}\mathcal{PT}\ket{\Psi}|^2)/2$,  where $\mathcal{T}=U\mathcal{K}$.   }
\label{fig2}
\end{figure}

%============================================================%
%\section{Conclusion}\label{sec:con}
%============================================================%

{\it Conclusion.--}\label{sec:con}
To summarize, here we have developed the resource theory of a special kind of antiunitary asymmetry, namely, the
${\cal PT}$-asymmetry with ${\cal P}$ being any
self-inverse unitary and ${\cal T}$ being the time-reversal. We
have introduced the notion of ${\cal PT}$-symmetric states, ${\cal PT}$-covariant operations and ${\cal PT}$-asymmetry measures.
We give several interesting ${\cal PT}$-asymmetry measures which are induced from different distance measures. Most importantly, we have proved
the duality relations between the ${\cal PT}$-asymmetry measures, coherence and concurrence. This also gives new interpretations to quantum coherence and
entanglement, which are special ${\cal PT}$-symmetries in some sense, thus unifying two fundamental resources
of quantum world.
Furthermore, we have argued that via the embedding  quantum simulator, the ${\cal PT}$-asymmetry measures can be
calculated efficiently. Finally, we have discussed the ${\cal PT}$-symmetric dynamics and proposed several
open problems. Our findings will open up new ways of thinking about quantum coherence and entanglement from another resource theoretic
point of view, i.e., the ${\cal PT}$-asymmetry measures.
The ${\cal PT}$-asymmetry is a just special kind of antiunitary asymmetry resource theory and
may pave the way to a general antiunitary asymmetry resource theory.

\smallskip
\noindent
\begin{acknowledgments}
K. B. thanks Yaobo Zhang for useful discussion.
This work is supported by the Natural Science Foundations of China (Grants No:11171301, No: 10771191 and No: 11571307) and
the Doctoral Programs Foundation of the Ministry of Education of China (Grant No. J20130061).

\end{acknowledgments}

\bibliographystyle{apsrev4-1}
 \bibliography{pt-sym-lit}

\appendix

\section{Relative entropy of $\mathcal{PT}$-asymmetry}\label{ap:Re_pt}
To prove the properties of $\Gamma_s$, we need the following lemma, which
is  not trivial, as $\mathcal{PT}$ is antilinear operator.
\begin{lem}\label{lem:1}
Given the self-inverse unitary operator $\mathcal{P}$ and time reversal operator $\mathcal{T}$,
for any two Hermitian operators Q and S, we have
\begin{eqnarray}\label{eq:eq1}
\Tr{Q\mathcal{PT}S\mathcal{PT}}=\Tr{\mathcal{PT}Q\mathcal{PT} S}.
\end{eqnarray}
\end{lem}

\begin{proof}
Due to the spectral decomposition
Theorem, Q and S can be written as
$Q=\sum_i\lambda_i\proj{\psi_i}$ and
$S=\sum_j\mu_j \proj{\phi_j}
$, respectively.
Thus, we have
\begin{eqnarray*}
&&\Tr{Q\mathcal{PT}S\mathcal{PT}}\\
&=&\sum_{ij}\lambda_i\mu_j\Tr{\proj{\psi_i}\mathcal{P}U\mathcal{K}\proj{\phi_j}\mathcal{K}U^\dag \mathcal{P}}\\
&=&\sum_{ij}\lambda_i\mu_j|\bra{\psi_i}\mathcal{P}U\ket{\phi^*_j}|^2,
\end{eqnarray*}
where the first equality comes from the fact that
$\mathcal{T}=U\mathcal{K}$, where $U$ is a unitary operator with $U=U^t$ and $\mathcal{K}=*$.
Similarly, we have
\begin{eqnarray*}
\Tr{\mathcal{PT}Q\mathcal{PT}S}
=\sum_{ij}\lambda_i\mu_j|\bra{\phi_j}\mathcal{P}U\ket{\psi^*_i}|^2.
\end{eqnarray*}
Hence, to prove $\Tr{Q\mathcal{PT}S\mathcal{PT}} =\Tr{\mathcal{PT}Q\mathcal{PT}S}$,  we only need to prove
that for any two pure states $\ket{\psi}$ and $\ket{\phi}$
\begin{eqnarray}
|\bra{\psi}V\ket{\phi^*}|=|\bra{\phi}V\ket{\psi^*}|,
\end{eqnarray}
where $V=\mathcal{P}U$ is a unitary operator.
Moreover,
since $\mathcal{P}=\mathcal{TPT}=U\mathcal{K}\mathcal{P}U\mathcal{K}=U\mathcal{P}^*U^*$,
then $P^*U^*=U^\dag \mathcal{P}=U^\dag \mathcal{P}^\dag$, which
implies that $U^t\mathcal{P}^t=\mathcal{P}U$. That is,
$V^t=V$. Therefore,
$
|\bra{\psi}V\ket{\phi^*}|
=|\bra{\phi^*}V^\dag \ket{\psi}|
=|\bra{\phi}V\ket{\psi^*}|
$

\end{proof}

\begin{prop}\label{prop:re1}
Given the self-inverse unitary operator $\mathcal{P}$ and time reversal operator $\mathcal{T}$,
let $\Gamma_r(\rho, \mathcal{PT})=\min_{\sigma\in Sym(\mathcal{P},\mathcal{T})}S(\rho||\sigma)$, then we have
\begin{eqnarray}\label{eq:re1}
\Gamma_r(\rho, \mathcal{PT})=S(\rho||\rho^{\mathcal{PT}})=S(\rho^{\mathcal{PT}})-S(\rho),
\end{eqnarray}
where $\rho^{\mathcal{PT}}=\frac{1}{2}(\rho+\mathcal{PT}\rho \mathcal{PT})$ is $\mathcal{PT}$-symmetric.
\end{prop}
\begin{proof}
Since it involves the complex conjugation $\mathcal{K}$, taking trace may be
complicated. As  for any linear operator A, $\mathcal{K}A\mathcal{K}=A^*$ and $\Tr{A^*}\neq \Tr{A}$
in general, thus we need be more careful to deal with taking trace here.
However, due to Lemma \ref{lem:1},  for any two Hermitian operators Q and S, we have
\begin{eqnarray}\label{eq:eq1}
\Tr{Q\mathcal{PT}S\mathcal{PT}}=\Tr{\mathcal{PT}Q\mathcal{PT} S}.
\end{eqnarray}
Then we follow the approach in \cite{Gour2009,Marvian2012} to complete the
proof. Due to the fact that $S(\rho||\sigma)\geq 0$ and
$S(\rho||\sigma)=0$ iff $\sigma=\rho$, thus
\begin{eqnarray}\label{ineq:ine1}
\Tr{\rho\log\sigma}\leq \Tr{\rho\log \rho}
\end{eqnarray}
wehre the equality holds iff $\sigma=\rho$, which
implies that $\max_{\sigma}\Tr{\rho\log\sigma}=\Tr{\rho\log\rho}$.

First, for any $\mathcal{PT}$-symmetric state $\sigma$,  we have
\begin{eqnarray*}
\Tr{\mathcal{PT}\rho \mathcal{PT} \log \sigma}
=\Tr{\rho \mathcal{PT}\log\sigma \mathcal{PT}}
=\Tr{\rho \log\sigma},
\end{eqnarray*}
where the first equality comes from  the  equation \eqref{eq:eq1} and
the second equality comes from the fact that $\sigma$ is $\mathcal{PT}$-symmetric.
Hence, we have
\begin{eqnarray}\label{eq:eq2}
&&\Tr{\rho\log\sigma}\nonumber\\
&=&\frac{1}{2}\Tr{\rho\log\sigma}+\frac{1}{2}\Tr{\mathcal{PT}\rho \mathcal{PT}\log\sigma}\nonumber\\
&=&\Tr{\rho^{\mathcal{PT}}\log\sigma},
\end{eqnarray}
where $\rho^{\mathcal{PT}}=\frac{1}{2}(\rho+\mathcal{PT}\rho \mathcal{PT})$. Thus, for the  $\mathcal{PT}$-symmetric state $\rho^\mathcal{{PT}}$,
$S(\rho||\rho^{\mathcal{PT}})=\Tr{\rho\log\rho}-\Tr{\rho\log\rho^{\mathcal{PT}}}=\Tr{\rho\log\rho}-\Tr{\rho^{\mathcal{PT}}\log\rho^{\mathcal{PT}}}
=S(\rho^{\mathcal{PT}})-S(\rho)$.

Next, we show that $\min_{\sigma\in Sym(\mathcal{P},\mathcal{T})}S(\rho||\sigma)=S(\rho||\rho^{\mathcal{PT}})$.

\begin{eqnarray*}
&&\min_{\sigma\in Sym(\mathcal{P},\mathcal{T})}S(\rho||\sigma)\\
&=&\min_{\sigma\in Sym(\mathcal{P},\mathcal{T})}[\Tr{\rho\log\rho}-\Tr{\rho\log\sigma}]\\
&=&\Tr{\rho\log\rho}-\max_{\sigma\in Sym(\mathcal{P},\mathcal{T})}\Tr{\rho\log\sigma}\\
&=&\Tr{\rho\log\rho}-\max_{\sigma\in Sym(\mathcal{P},\mathcal{T})}\Tr{\rho^{\mathcal{PT}}\log \sigma}\\
&=&\Tr{\rho\log\rho}-\Tr{\rho^{\mathcal{PT}}\log \rho^{\mathcal{PT}}}\\
&=&S(\rho^{\mathcal{PT}})-S(\rho),
\end{eqnarray*}
where the third equality comes from \eqref{eq:eq2} and the fourth equality comes from
the  inequality \eqref{ineq:ine1}. Therefore, the proof of the proposition is completed.

\end{proof}

\begin{prop}
Given the self-inverse unitary operator $\mathcal{P}$ and time reversal operator $\mathcal{T}$, then
$\Gamma_r$ satisfy the conditions (C1), (C2), (C2') and (C3), so it is
a proper $\mathcal{PT}$-asymmetry measure.
\end{prop}
\begin{proof}
Since $S(\rho||\sigma)=0$ iff $\rho=\sigma$, thus $\Gamma_r$ satisfy (C1).
Besides, as the relative entropy is contracted under quantum operations \cite{Lindblad1975} and
jointly convex \cite{Ruskai2002}, then $\Gamma_r$ satisfy the conditions (C2) and (C3). Moreover,
we use the  the techniques in Ref.\cite{Marvian2012,Baumgratz2014} to prove that $\Gamma_r$ satisfy the condition (C2').

Take a special self-inverse unitary operator $\mathcal{P}_0=I$ and a time reversal operator
$\mathcal{T}_0=*$, then it is easy to see that  there exist a set
of orthonormal  pure states $\set{\ket{\mu}}_\mu$  with $\proj{\mu}\in Sym(\mathcal{P}_0, \mathcal{T}_0)$.
For any selective $\mathcal{PT}$-covariant operation $\Phi$ with $K_\mu(\mathcal{PT}(\cdot)\mathcal{PT})K^\dag_\mu=\mathcal{PT}K_\mu(\cdot)K^\dag_\mu \mathcal{PT}$ for any $\mu$,
it is easy to verify that the quantum operations $\widetilde{\Phi}$ with
Kraus operators $\widetilde{K}_\mu=\ket{\mu}\ot K_\mu$ is selective
$\mathcal{PT}$-covariant with respect to $(\mathcal{P},\mathcal{T})$ and $(\mathcal{P}_0\ot \mathcal{P}, \mathcal{T}_0\ot \mathcal{T})$, that is, $\widetilde{K}_\mu \mathcal{PT}\rho \mathcal{PT} \widetilde{K}^\dag_\mu=\mathcal{PT}\widetilde{K}_\mu \rho \widetilde{K}^\dag_\mu \mathcal{PT} $.
Thus $\widetilde{\Phi}(\rho)=\sum_\mu p_{\mu}
\proj{\mu}\ot \rho_\mu$, where $\rho_\mu=K_\mu\rho K^\dag_\mu/p_\mu$ with $p_\mu=\mathrm{Tr}(K_\mu \rho K^\dag_\mu)$.
As we have proved that $\Gamma_r$ satisfies the condition (C2), which implies that
$\Gamma_r(\sum_\mu p_{\mu}
\proj{\mu}\ot \rho_\mu,\mathcal{P}_0\mathcal{T}_0\ot \mathcal{PT})\leq \Gamma_r(\rho, \mathcal{PT})$.
Moreover,
\begin{eqnarray}
&&\Gamma_r(\sum_\mu p_{\mu}\proj{\mu}\ot \rho_\mu,\mathcal{P}_0\mathcal{T}_0\ot \mathcal{PT}\nonumber)\\
&=&S(\sum_\mu p_{\mu}\proj{\mu}\ot \rho^{\mathcal{PT}}_\mu)-S(\sum_\mu p_{\mu}\proj{\mu}\ot \rho_\mu)\nonumber\\
&=&\sum_\mu p_\mu S(\rho^{\mathcal{PT}}_\mu)-\sum_\mu p_{\mu}S(\rho_{\mu})\nonumber\\
&=&\sum_\mu p_\mu \Gamma_r(\rho_\mu,\mathcal{PT}).
\end{eqnarray}
where the first equality comes from the Proposition \ref{prop:re1} and
the second equality comes from the following  fact:
\begin{eqnarray}
S(\sum_ip_i\proj{i}\ot\rho_i)=\sum_ip_iS(\rho_i)+H(\set{p_i}),
\end{eqnarray}
where $\set{\ket{i}}$ is a set of othornormal pure sates and $H(\set{p_i})=
\sum_i-p_i\log p_i$ is the Shannon entropy of the probability distribution $\set{p_i}$.
Therefore, $\Gamma_r$ satisfies the condition $(C2')$, that is,  $\Gamma_r$
is a proper $\mathcal{PT}$-asymmetry measure.

\end{proof}

\begin{lem}
The measure $\Gamma_r$ is additive, that is
\begin{eqnarray}
\Gamma_r(\rho\ot \sigma, \mathcal{PT}\ot \mathcal{P}'\mathcal{T}')
=\Gamma_r(\rho, \mathcal{PT})+\Gamma_r(\sigma, \mathcal{P}'\mathcal{T}') \nonumber\\
\end{eqnarray}
\end{lem}
\begin{proof}
This comes directly from the representation \eqref{eq:re1} of $\Gamma_r$.
\end{proof}

\begin{lem}
The relative entropy of $\mathcal{PT}$-asymmetry is asymptotically continuous, i.e., for
$\norm{\rho-\sigma}_1\leq \epsilon\leq 1/e$, then
\begin{eqnarray}
\abs{\Gamma_r(\rho, \mathcal{PT})-\Gamma_r(\sigma, \mathcal{PT})}\leq 2\epsilon \log d +2\eta(\epsilon),
\end{eqnarray}
where d is the dimension of the system and $\eta(\epsilon)=-\epsilon \log \epsilon$.
\end{lem}
\begin{proof}
Due to  $\norm{\rho-\sigma}_1\leq \epsilon$, we have
$\norm{\rho^{\mathcal{PT}}-\sigma^{\mathcal{PT}}}_1\leq \epsilon$. Then by
Fannes' inequality \cite{Fannes1973} and \eqref{eq:re1}, we get the desired result.
\end{proof}

\section{Skew information of $\mathcal{PT}$-asymmetry } \label{ap:Sk_pt}

\begin{prop}\label{prop:skewm}
Given the self-inverse unitary operator $\mathcal{P}$ and time reversal operator $\mathcal{T}$, then
$\Gamma_s$ satisfy the conditions (C1), (C2), (C2') and (C3), so it is
a proper $\mathcal{PT}$-asymmetry measure.
\end{prop}
\begin{proof}
If $[\rho, \mathcal{PT}]=0$, then obviously $\Gamma_s(\rho)=0$. So we only need to
prove the inverse direction.
Since $\mathcal{PT}\rho^{1/2}\mathcal{PT}=\mathcal{P}U\mathcal{K}\rho^{1/2} \mathcal{K}U^* \mathcal{P}=\mathcal{P}U(\rho^{1/2})^* U^\dag \mathcal{P}$ is positive  and
$(\mathcal{PT}\rho^{1/2}\mathcal{PT})^2=\mathcal{PT}\rho \mathcal{PT}$, then $\mathcal{PT}\rho^{1/2}\mathcal{PT}$ is the square root of the
quantum state $\mathcal{PT}\rho \mathcal{PT}$. Besides, $\Tr{\rho^{1/2}\mathcal{PT}\rho^{1/2}\mathcal{PT}}\leq \mathrm{Tr}(|\rho^{1/2}\mathcal{PT}\rho^{1/2}\mathcal{PT}|)=F(\rho, \mathcal{PT}\rho \mathcal{PT})$, where for any two states
$\rho_1$ and $\rho_2$, $F(\rho_1,\rho_2):=\Tr{|\rho^{1/2}_1\rho^{1/2}_2|}$. That is,
$\Gamma_s(\rho, \mathcal{PT})\geq 1-F(\rho, \mathcal{PT}\rho \mathcal{PT})$. As $\Gamma_s(\rho, \mathcal{PT})=0$, then $F(\rho, \mathcal{PT}\rho \mathcal{PT})=1$
which means $\rho=\mathcal{PT}\rho \mathcal{PT}$. Thus $\Gamma_s$ satisfy the condition (C1).

Since $\mathcal{PT}\rho^{1/2}\mathcal{PT}$ is the square root of the
quantum state $\mathcal{PT}\rho \mathcal{PT}$, then
$\Gamma_s(\rho,\mathcal{PT})=1-\Tr{\rho^{1/2}(\mathcal{PT}\rho \mathcal{PT})^{1/2}}$.
The convexity of $\Gamma_s$:
\begin{eqnarray*}
\Gamma_s(p\rho+(1-p)\sigma, \mathcal{PT})\leq p\Gamma_s(\rho, \mathcal{PT})+(1-p)\Gamma_s(\sigma, \mathcal{PT})
\end{eqnarray*}
comes from the following famous result \cite{Lieb1973}:
\begin{eqnarray*}
&&\Tr{(p\rho_1+(1-p)\rho_2)^\alpha (p\sigma_1+(1-p)\sigma_2)^{1-\alpha}}\\
&\geq& p \Tr{\rho^\alpha_1\sigma^{1-\alpha}_1}+(1-p)\Tr{\rho^\alpha_2\sigma^{1-\alpha}_2}.
\end{eqnarray*}
where $\rho_1,\rho_2,\sigma_1,\sigma_2$ are quantum states and $p,\alpha\in[0,1]$.

Besides, for the quantity $D_\alpha(\rho_1,\rho_2)=\Tr{\rho^{\alpha}_1\rho^{1-\alpha}_2}$ with
$\alpha\in (0,1)$, it holds that \cite{Hayashi2006}:
\begin{eqnarray}
D_\alpha(\rho_1,\rho_2)\leq D_\alpha(\Phi(\rho_1),\Phi(\rho_2))
\end{eqnarray}
where $\Phi$ is a quantum operation.
Take $\alpha=1/2$, then we will find that $\Gamma_s$ satisfies the condition (C2).

Finally, similar to proof in Proposition 2,
we take a special self-inverse unitary operator $\mathcal{P}_0=I$ and time reversal operator
$\mathcal{T}_0=*$ with a set
of orthonormal  pure states $\set{\ket{\mu}}_\mu$, $\proj{\mu}\in Sym(\mathcal{P}_0, \mathcal{T}_0)$.
For any selective $\mathcal{PT}$-covariant operation $\Phi$ with $K_\mu(\mathcal{PT}(\cdot)\mathcal{PT})K^\dag_\mu=\mathcal{PT}K_\mu(\cdot)K^\dag_\mu \mathcal{PT}$ for any $\mu$,
it is easy to verify that the quantum operations $\widetilde{\Phi}$ with
Kraus operators $\widetilde{K}_\mu=\ket{\mu}\ot K_\mu$ is selective
$\mathcal{PT}$-covariant with respect to $(\mathcal{P}, \mathcal{T})$ and $(\mathcal{P}_0\ot \mathcal{P}, \mathcal{T}_0\ot \mathcal{T})$,  thus $\widetilde{\Phi}(\rho)=\sum_\mu p_{\mu}
\proj{\mu}\ot \rho_\mu$, where $\rho_\mu=K_\mu\rho K^\dag_\mu/p_\mu$ with $p_\mu=\mathrm{Tr}(K_\mu \rho K^\dag_\mu)$.
As we have proved that $\Gamma_r$ satisfies the condition (C2), which implies that
$\Gamma_s(\sum_\mu p_{\mu}
\proj{\mu}\ot \rho_\mu, \mathcal{P}_0\mathcal{T}_0\ot \mathcal{PT})\leq \Gamma_s(\rho, \mathcal{PT})$.
Moreover, it is easy to verify that
 \begin{eqnarray}
&&\Gamma_s(\sum_\mu p_{\mu}\proj{\mu}\ot \rho_\mu, \mathcal{P}_0\mathcal{T}_0\ot \mathcal{PT}\nonumber)\\
&=&\sum_\mu p_\mu \Gamma_s(\rho_\mu, \mathcal{PT}).
\end{eqnarray}
Hence, the condition (C2') holds for $\Gamma_s$.

\end{proof}

\section{Geometric measure of $\mathcal{PT}$ asymmetry}\label{ap:F_pt}
\begin{prop}
Given the self-inverse unitary operator $\mathcal{P}$ and time reversal operator $\mathcal{T}$, then
$\Gamma_F$ satisfy the conditions (C1), (C2), (C2') and (C3), so it is
a proper $\mathcal{PT}$-asymmetry measure.
\end{prop}

\begin{proof}
(C1) is obvious, since $F(\rho, \mathcal{PT}\rho \mathcal{PT})=1$ iff $\rho=\mathcal{PT}\rho \mathcal{PT}$.
The convexity (C3) of $\Gamma_F$ comes from the joint concavity of fidelity (See \cite{Miszczak2009} and the reference therein ).
As fidelity is non-decreasing under CPTP maps (See \cite{Miszczak2009} and the reference therein ), thus
$\Gamma_F$ satisfies the condition (C2).
Moreover, using a similar method as the proof in  $\Gamma_r$ and
$\Gamma_s$, we can prove the condition (C2').
Take a special self-inverse unitary operator $\mathcal{P}_0=I$ and a time reversal operator
$\mathcal{T}_0=*$ with a set
of orthonormal  pure states $\set{\ket{\mu}}_\mu$, $\proj{\mu}\in Sym(\mathcal{P}_0, \mathcal{T}_0)$.
For any selective $\mathcal{PT}$-covariant operation $\Phi$ with $K_\mu(\mathcal{PT}(\cdot)\mathcal{PT})K^\dag_\mu=\mathcal{PT}K_\mu(\cdot)K^\dag_\mu \mathcal{PT}$ for any $\mu$,
it is easy to verify that the quantum operations $\widetilde{\Phi}$ with
Kraus operators $\widetilde{K}_\mu=\ket{\mu}\ot K_\mu$ is selective
$\mathcal{PT}$-covariant with respect to $(\mathcal{P}, \mathcal{T})$ and $(\mathcal{P}_0\ot \mathcal{P}, \mathcal{T}_0\ot \mathcal{T})$, thus $\widetilde{\Phi}(\rho)=\sum_\mu p_{\mu}
\proj{\mu}\ot \rho_\mu$, where $\rho_\mu=K_\mu\rho K^\dag_\mu/p_\mu$ with $p_\mu=\mathrm{Tr}(K_\mu \rho K^\dag_\mu)$.
As we have proved that $\Gamma_r$ satisfies the condition (C2), which implies that
$\Gamma_F(\sum_\mu p_{\mu}
\proj{\mu}\ot \rho_\mu,\mathcal{P}_0\mathcal{T}_0\ot \mathcal{PT})\leq \Gamma_F(\rho, \mathcal{PT})$.
Moreover, it is easy to verify that
 \begin{eqnarray}
&&\Gamma_F(\sum_\mu p_{\mu}\proj{\mu}\ot \rho_\mu,\mathcal{P}_0\mathcal{T}_0\ot \mathcal{PT}\nonumber)\\
&=&\sum_\mu p_\mu \Gamma_F(\rho_\mu, \mathcal{PT}).
\end{eqnarray}
Hence, the condition (C2') holds for $\Gamma_F$.

\end{proof}

\begin{lem}\label{lem:prop_of_fid}
The $\mathcal{PT}$ asymmetry measure $\Gamma_F$ is continuous, that is
\begin{eqnarray}
|\Gamma_F(\rho, \mathcal{PT})
-\Gamma_F(\sigma, \mathcal{PT})|
\leq 2\sqrt{\norm{\rho-\sigma}_1}.
\end{eqnarray}

\end{lem}
\begin{proof}
\begin{eqnarray*}
&&\abs{\Gamma_F(\rho,\mathcal{PT})
-\Gamma_F(\sigma,\mathcal{PT})}\\
&=&\abs{F(\rho, \mathcal{PT}\rho \mathcal{PT})-F(\sigma, \mathcal{PT}\sigma \mathcal{PT})}\\
&\leq&\abs{F(\rho, \mathcal{PT}\rho \mathcal{PT})-F(\rho, \mathcal{PT}\sigma \mathcal{PT})}\\
&&+\abs{F(\rho, \mathcal{PT}\sigma \mathcal{PT})-F(\sigma, \mathcal{PT}\sigma \mathcal{PT})}\\
&\leq& \sqrt{1-F(\mathcal{PT}\rho \mathcal{PT}, \mathcal{PT}\sigma \mathcal{PT})^2}
+\sqrt{1-F(\rho,\sigma)^2}\\
&=&2\sqrt{1-F(\rho,\sigma)^2}
\leq 2\sqrt{2}\sqrt{1-F(\rho,\sigma)}\\
&\leq& 2\sqrt{\norm{\rho-\sigma}_1}.
\end{eqnarray*}
where the second inequality comes from \cite{Fannes2012} and the last
inequality comes from the  Fuchs-van de Graaf inequality \cite{Fuchs1999}.
\end{proof}

\section{Duality of $\mathcal{PT}$-Asymmetry and Entanglement}\label{ap:PTvsEnt}

\begin{thm}\label{prop:PTvsEnt}
Given a two-qubit system with the self-inverse unitary operator $\mathcal{P}=\sigma_y\ot \sigma_y$
and time reversal operator $\mathcal{T}=*$, for pure bipartite states $\ket{\Psi}$ we have
\begin{eqnarray}\label{eq:skvse}
\Gamma_s(\Psi, \mathcal{PT})+C(\Psi)^2=1,\\
\label{eq:fvse}
\Gamma_F(\Psi, \mathcal{PT})+C(\Psi)=1,
\end{eqnarray}
and
\begin{eqnarray}\label{eq:rvse}
\Gamma_r (\Psi, \mathcal{PT})=H\Pa{\frac{1}{2}-\frac{1}{2}C(\Psi)}
\end{eqnarray}
where $H(p)=-p\log(p)-(1-p)\log(1-p)$ is Shannon entropy for the probability
distribution $\set{p,1-p}$ and $C(\Psi)$ is the concurrence for pure state $\Psi$.

For any two-qubit states $\rho$, the equalities  may not hold. However, we still have the
following inequality:
\begin{eqnarray}
\label{ineq:skvse}\Gamma_s(\rho, \mathcal{PT})+C(\rho)^2\leq1,\\
\label{ineq:fvse}\Gamma_F(\rho, \mathcal{PT})+C(\rho)\leq 1,\\
\label{ineq:rvse} \Gamma_r(\rho, \mathcal{PT})\leq H(\frac{1}{2}-\frac{1}{2}C(\rho)),
\end{eqnarray}
where $C(\rho)=\min \sum_k p_k C(\Psi_k)$ and
the minimum is taken over all the pure states decomposition of $\rho=\sum_k p_k\proj{\Psi_k}$ \cite{Wootters1998,Uhlmann2000}.

In fact,
\begin{eqnarray}\label{eq:coa}
\Gamma_F(\rho, \mathcal{PT})+CoA(\rho)=1
\end{eqnarray}
where the concurrence of assistance $CoA(\rho)=\max \sum_k p_k C(\Psi_k)$ and the maximum is taken over all the
 pure states decomposition of $\rho=\sum_k p_k\proj{\Psi_k}$ \cite{Laustsen2003,Gour2005}.

\end{thm}
\begin{proof}
The equations \eqref{eq:skvse}, \eqref{eq:fvse}, \eqref{eq:rvse} come directly from
\eqref{eq:psk}, \eqref{eq:pre} and \eqref{eq:defc} when $\mathcal{P}=\sigma_y\ot \sigma_y$ and $\mathcal{T}=*$.
Thus we only need to verify that $\mathcal{P}=\sigma_y\ot \sigma_y$ and $\mathcal{T}=*$ satisfy these three condition.
Obviously, $\sigma_y\ot \sigma_y=(\sigma_y\ot \sigma_y)^\dag$ and $(\sigma_y\ot \sigma_y)^2=I=\mathcal{T}^2$.
Furthermore, as $\mathcal{K}\sigma_y \mathcal{K}=\sigma^*_y=-\sigma_y$, then $\mathcal{K}(\sigma_y\ot\sigma_y)\mathcal{K}=\sigma_y\ot\sigma_y$.

Due to the convexity of $\Gamma_s$ and $C(\rho)$, thus for any pure state
decomposition of $\rho=\sum_k p_k\proj{\Psi_k}$, we have
$\Gamma_s(\rho, \mathcal{PT})\leq \sum_k p_k \Gamma_s(\Psi_k,\mathcal{PT})$ and
$C(\rho)\leq \sum_k p_k C(\Psi_k)$. Based on the convexity of
$f(x)=x^2$, we have $C(\rho)^2\leq \sum_k p_kC(\Psi_k)^2$.
Therefore, due to the equality \eqref{eq:skvse},
we get the inequality \eqref{ineq:skvse}.

Similarly, \eqref{ineq:fvse} comes directly from the convexity of $\Gamma_F$ and $C(\rho)$.
And \eqref{ineq:rvse} also comes from the convexity of $\Gamma_r$, $C(\rho)$ and the
concavity of Shannon entropy $H$,
as for any pure states decomposition of $\rho=\sum_k p_k\proj{\Psi_k}$,
\begin{eqnarray*}
\Gamma_r(\rho, \mathcal{PT})&\leq& \sum_k p_k\Gamma_r(\Psi_k, \mathcal{PT})\\
&=&\sum_k p_k H(\frac{1}{2}-\frac{1}{2}C(\Psi_k))\\
&\leq& H(\frac{1}{2}-\frac{1}{2}\sum_k p_k C(\Psi_k))\\
&\leq& H(\frac{1}{2}-\frac{1}{2}C(\rho)).
\end{eqnarray*}

Finally, \eqref{eq:coa} comes directly from the definition of
concurrence of assistance (CoA) and Eq.\ref{eq:deco_F}.

\end{proof}

%============================================================%
\section{$\mathcal{PT}$-symmetric dynamics}\label{sec:discus}
%============================================================%
Consider the unitary operator $V$ which is $\mathcal{PT}$-covariant, that is
\begin{eqnarray}
\mathcal{PT}V\rho V^\dag \mathcal{PT}=V\mathcal{PT}\rho \mathcal{PT} V^\dag, ~for~any~\rho,
\end{eqnarray}
then
\begin{eqnarray}
V\mathcal{PT}=e^{i\theta} \mathcal{PT}V.
\end{eqnarray}

Besides, the unitary operator $V$ is called $\mathcal{PT}$-invariant unitary if $[V,\mathcal{PT}]=0$.
If $[V, \mathcal{PT}]=0$, then $[V^\dag, \mathcal{PT}]=0$. Note that, if the Hamiltonian $H$ satisfying
$\set{H, \mathcal{PT}}=0$ where $\set{A, B}=AB+BA$, then the unitaries  $e^{iHt}$ satisfy$[e^{iHt}, \mathcal{PT}]=0$.

Similar to entanglement \cite{Nielsen1999} and coherence \cite{Du2015}, we can also consider the  state transformation
under $\mathcal{PT}$-covariant operation.

\begin{prop}\label{prop:con_tran}
Pure state $\ket{\psi}$ can be transformed to $\ket{\phi}$ under selective $\mathcal{PT}$-covariant
operations if and only if $|\bra{\psi}\mathcal{PT}\ket{\psi}|\leq |\bra{\phi}\mathcal{PT} \ket{\phi}|$.
\end{prop}
\begin{proof}
The "only if" part is obvious, we only need to prove the "if" part.
Define $\cH_{\mathcal{PT}}=\set{\ket{\psi}\in \cH: \mathcal{PT} \ket{\psi}=\ket{\psi}}$,
then $\cH_{\mathcal{PT}}$ is a real Hilbert space with $dim \cH_{\mathcal{PT}}=dim \cH$ \cite{Uhlmann2000,Uhlmann2016}.
Thus the basis of $\cH_{\mathcal{PT}}$  is also a basis of $\cH$, which is called $\mathcal{PT}$-invariant basis. Here we use
$\set{\ket{i}_O}^d_{i=1}$ and $\set{\ket{i}_N}^d_{i=1}$ to denote the initial given basis of $\cH$ and the
 $\mathcal{PT}$-invariant basis,  $\ket{\psi}_O$ and $\ket{\psi}_N$ to denote the representation of state $\ket{\psi}$ in the
 basis $\set{\ket{i}_O}^d_{i=1}$ and $\set{\ket{i}_N}^d_{i=1}$, respectively.
 Obviously, there exists a unitary operator $U$ such that $\ket{\psi}_N=U\ket{\psi}_O$
 for any $\psi$.
Besides, the ${\cal PT}$ operator in the new  basis
 $\set{\ket{i}_N}^d_{i=1}$ is equivalent to the complex conjugation with respect to  this basis, denoted by $\mathcal{K}_{N}$.  Then
 ${\cal PT}\ket{\psi}_O=\mathcal{K}_N\ket{\psi}_N$
 and $\bra{\psi}{\cal PT}\ket{\psi}\equiv\bra{\psi}{\cal PT}\ket{\psi}_O=\bra{\psi}\mathcal{K}_N\ket{\psi}_N=
 \bra{\psi}U^\dag\mathcal{K}_NU\ket{\psi}_O$, which implies that ${\cal PT}=U^\dag{\cal K}_N U$.
Thus $|\bra{\psi}\mathcal{PT}\ket{\psi}|\leq |\bra{\phi}\mathcal{PT} \ket{\phi}|$ is equivalent to
$|\bra{\psi}\mathcal{K}_N\ket{\psi}_N|\leq |\bra{\phi}\mathcal{K}_N \ket{\phi}_N|$. Due to
 Ref.\cite{Gour2009JMP}, there exists a CPTP map $\Phi$ with Kraus operators
 $\set{K_\mu}$ such that $[K_\mu, \mathcal{K}_N]=0$ and $\Phi(\proj{\psi}_N)=\proj{\phi}_N$.
 Therefore, in the  initial given basis  $\set{\ket{i}_O}^d_{i=1}$,  the quantum operation $\widetilde{\Phi}$ with $\widetilde{K}_\mu=U^\dag K_\mu U$ satisfy the conditions $[\widetilde{K}_\mu, {\cal PT}]=0$ and $\widetilde{\Phi}(\proj{\psi})=\proj{\phi}$.

\end{proof}

The ${\cal PT}$-asymmetry measure introduced here is connected with the time reversal symmetry monotone in \cite{Gour2009JMP} up to a unitary. However,
 it is  hard to find  the unitary $U$ such that $\Gamma(\rho, \mathcal{PT})=\Gamma(U\rho U^\dag,\mathcal{K}_N)$ where $\mathcal{K}_N$ is the complex conjugation
   with the $\mathcal{PT}$-invariant basis. And such a unitary, especially a global unitary on the composite system,
may ruin the nonlocality of ${\cal PT}$ -symmetry.

Except the $\mathcal{PT}$-covariant operations,  we can also consider
the operations which map the $\mathcal{PT}$-symmetric state to $\mathcal{PT}$-symmetric state, that is
$\Phi(Sym(\mathcal{P}, \mathcal{T}))\subset Sym(\mathcal{P}, \mathcal{T})$. Such operations are called $\mathcal{PT}$-preserving operations.
Obviously, all $\mathcal{PT}$-covariant operations are  $\mathcal{PT}$-preserving operations.
Moreover, for any operations $\Phi=\sum_\mu K_\mu(\cdot)K^\dag_\mu$ with $K_\mu(Sym(\mathcal{P}, \mathcal{T}))K^\dag_\mu\subset Sym(\mathcal{P}, \mathcal{T})$,
we call such operation selective $\mathcal{PT}$-preserving operations, which is similar to
the definition of incoherent operations \cite{Baumgratz2014}. We can weaken the conditions
(C2) and (C2') to the following conditions:

(C2a) Monotone under $\mathcal{PT}$-preserving operations $\Phi_{\mathcal{PT}_{pre}}$, that is
$\Gamma(\Phi_{\mathcal{PT}_{pre}}(\rho), \mathcal{PT})\leq \Gamma(\rho,\mathcal{PT})$.

(C2'a) Monotone under selective $\mathcal{PT}$-preserving operations, that is
 $\sum_\mu p_\mu \Gamma(\rho_\mu, \mathcal{PT})\leq \Gamma(\rho, \mathcal{PT})$,
where $K_\mu Sym(\mathcal{P}, \mathcal{T}) K^\dag_\mu\subset Sym(\mathcal{P}, \mathcal{T})$ and
$\rho_\mu=K_\mu\rho K^\dag_\mu/p_\mu$ with $p_\mu=\mathrm{Tr}(K_\mu \rho K^\dag_\mu)$.

Therefore, we can also consider the $\mathcal{PT}$-asymmetry monotone which
satisfy the conditions (C1), (C2a), (C2a') and (C3),  which will be
explored in a future work.

\end{document}